\documentclass{article}
\usepackage{times}
\usepackage{xcolor}
\usepackage{soul}
\usepackage[utf8]{inputenc}
\usepackage[small]{caption}
\usepackage{natbib}

\usepackage{float}
\usepackage{amsmath}
\usepackage{amssymb}
\usepackage{amsthm}
\usepackage{subcaption}
\newcommand{\R}{\mathcal{R}}
\newcommand{\leaveout}[1]{}
\usepackage[noend]{algpseudocode}

\newtheorem{prop}{Proposition}
\newtheorem{theorem}{Theorem}
\usepackage{mathtools}
\DeclarePairedDelimiter{\ceil}{\lceil}{\rceil}
\DeclarePairedDelimiter\floor{\lfloor}{\rfloor}

\usepackage{verbatim} 
\usepackage{algorithm}
\usepackage{algpseudocode}

\usepackage[colorinlistoftodos]{todonotes}
\usepackage{xcolor}
\newcommand\myworries[1]{\textcolor{red}{#1}}

\allowdisplaybreaks

\begin{document}
\title{Adversarial Task Assignment}
\author{
	Chen Hajaj and  
	Yevgeniy Vorobeychik
	\\ 
	Electrical Engineering and Computer Science\\Vanderbilt University\\Nashville, TN\\
	\{chen.hajaj,
	yevgeniy.vorobeychik\}@vanderbilt.edu
}
\date{}
\maketitle

\begin{abstract}
The problem of assigning tasks to workers is of long-standing fundamental importance. Examples of this include the classical problem of assigning computing tasks to nodes in a distributed computing environment, assigning jobs to robots, and crowdsourcing. Extensive research into this problem generally addresses important issues such as uncertainty and incentives. However, the problem of adversarial tampering with the task assignment process has not received as much attention. 

We are concerned with a particular adversarial setting where an attacker may target a set of workers in order to prevent the tasks assigned to these workers from being completed. 
When all tasks are homogeneous, we provide an efficient algorithm for computing the optimal assignment. When tasks are heterogeneous, we show that the
adversarial assignment problem is NP-Hard, and present an algorithm for solving it approximately. Our theoretical results are accompanied by extensive experiments showing the effectiveness of our algorithms.

\end{abstract}

\section{Introduction}
The problem of allocating a set of tasks among a collection of workers has been a fundamental research question in a broad array of domains, including
distributed computing, robotics, and, recently, crowdsourcing~\cite{alistarh2012allocate,stone1999task,liu2017sequential}.
Despite the extensive interest in the problem, however, there is little prior work on task assignment in settings where workers may be attacked.
Such \emph{adversarial task assignment} problems can arise, for example, when tasks are of high economic or political consequence, such as in robotic rescue missions following terror activities, or crowdsourcing to determine 
which executables are malicious or benign, or which news stories constitute fake news. 

We investigate the adversarial task assignment problem in which a rational external attacker targets one or more workers after tasks have already been assigned.
Equivalently, this can be viewed as a robust task assignment problem with unknown uncertainty about worker failures.
We formalize the interaction between the attacker and requester (defender) as a Stackelberg game in which the defender first chooses an assignment, and the attacker subsequently attacks a set of workers so as to maximize the defender's losses from the attack.
We seek a strong Stackelberg equilibrium (SSE) of this game and focus on computing an optimal robust assignment.

Our analysis begins with 
a setting in which tasks are homogeneous, that is, all tasks have the same utility for the defender (e.g., rescue soldiers from a battlefield, or label a large dataset of images). 
We characterize the optimal structure of a robust assignment, and use this insight to develop an algorithm that extracts this assignment in time linear in the number of tasks and targets, and quadratic in the number of workers.
We show that this algorithm significantly outperforms several baselines, and obtains a good solution even when no adversary is present.



Next, we turn to heterogeneous task settings. This case, it turns out, is considerably more challenging. 
Specifically, we show that it may be beneficial to assign more than a single worker to a task.
Moreover, even if we impose a restriction that only a single worker can be assigned to a task (optimal when tasks are homogeneous), extracting the optimal assignment is strongly NP-Hard. To overcome this issue, we  propose an integer programming approach for solving the restricted problem, as well as an algorithm for finding an approximately optimal assignment in the general case. 
Again, our experiments show that our approach significantly outperforms several baselines.
\paragraph{Related Work} 

The problem of task assignment in adversarial settings has been considered from several perspectives.
One major stream of literature is about robots acting in adversarial environments. 
Alighanbari and How~\cite{alighanbari2005cooperative} consider assigning weapons to targets, somewhat analogous to our problem, but do not model the decision of the adversary; their model also has rather different semantics than ours.
Robotic soccer is another common adversarial planning problem, although the focus is typically on coordination among multiple robots when two opposing teams are engaged in coordination and planning~\cite{jones2006dynamically}.

Another major literature stream which considers adversarial issues is crowdsourcing.
One class of problems is a number of workers to hire~\cite{carvalho2016many}, the issue of individual worker incentives in truthfully responding to questions~\cite{singla2013truthful}, or in the amount of effort they devote to the task~\cite{tran2014budgetfix,elmalech2016extending,liu2017sequential}, rather than adversarial reasoning per se.
Another, more directly adversarial setting, considers situations where some workers simply answer questions in an adversarial way~\cite{Ghosh11,Steinhardt16}.
However, the primary interest in this work is robust estimation when tasks are assigned randomly or exogenously, rather than task assignment itself.
Similarly, prior research on machine learning when a portion of data is adversarially poisoned~\cite{chen-icml11,xu2010,feng2014robust,chen2013robust,Liu17} focuses primarily on the robust estimation problem, and not task assignment; in addition, it does not take advantage of structure in the data acquisition process, where workers, rather than individual data points, are attacked. Other works \cite{gu2005robust,alon2015robust} focus on the change of the system after the assignment process and the structure of the social network rather than the assignment process itself.


Our work has a strong connection to the literature on Stackelberg security games~\cite{conitzer2006computing,korzhyk2010complexity,tambe2011security}.
However, the mathematical structure of our problem is quite different.
For example, we have no protection resources to allocate, and instead, the defender's decision is about assigning tasks to potentially untrusted workers.


\section{Model}
\label{sec:model}


Consider an environment populated with a single requester (hereafter denoted by ``\textit{defender}''), a set of $n$ workers, $W$, a set of $m$ tasks, $T$, and an \emph{adversary}.
Furthermore, each worker $w \in W$ is characterized by a capacity constraint $c_w$, which is the maximum number of tasks it can be assigned, and an individual proficiency or the probability of successfully completing a task, denoted by $p_{w}$.
Worker proficiencies are assumed to be common knowledge to both the defender and attacker.
Such proficiencies can be learned from experience~\cite{sheng2008get,dai2011artificial,manino2016efficiency}; moreover, in many settings, these are provided by the task assignment (e.g., crowdsourcing) platform, 
in the form of a reputation system~\cite{mason2012conducting}.

For exposition purposes, we index the workers by integers $i$ in decreasing order of their proficiency, so that $P=(p_1,\ldots,p_n) $ s.t. $ p_{i}\geq p_{j}$ $\forall i<j$, and denote the set of $ k $ most proficient workers by $W^k$. 
Thus, the capacity of worker $i$ would be denoted by $c_i$.
Each task $t \in T$ is associated with a utility $u_{t}$ that the defender obtains if this task is completed successfully.
If the task is not completed successfully, the defender obtains zero utility from it.

%

We focus on the common case where the defender faces a budget constraint of making at most $B \leq m$ assignments; the setting with $B > m$ necessitates different algorithmic techniques, and is left for future work.
The defender's fundamental decision is the \emph{assignment} of tasks to workers. 
Formally, an assignment $s$ specifies a subset of tasks $T'(s)$ and the set of workers, $W_t(s)$ assigned to each task $t \in T'(s)$.

Suppose that multiple workers are assigned to a task $t$, and let $L_t(s)$ denote the labels returned by workers in $W_t(s)$ for $t$ (for example, these could simply indicate whether a worker successfully complete the task).
Then the defender determines the final label to assign to $t$ (e.g., whether or not the task has been successfully completed) according to some deterministic mapping $ \delta: L_t(s) \to l $ (e.g., majority label), such that $ L\in \{1,\ldots,j_t\}^{|W_t(s)|} $ and $ l\in\{1,\ldots,j_t\} $.
Naturally, whenever a single worker $w$ is a assigned to a task and returns a label $l_w$, $\delta(l_w) = l_w$.
Let $\iota_t$ be the (unknown) correct label corresponding to a task $t$; this could be an actual label, such as the actual object in the image, or simply a constant 1 if we are only interested in successful completion of the task.
The defender's expected utility when assigning a set of tasks $T'(s)$ to workers and obtaining the labels is then
\begin{equation}
u_{def}(s)=\sum_{t\in T'(s)} u_{t} \Pr\{\delta(L_t(s)) = \iota_t\},
\label{eq:basic}
\end{equation}
where the probability
is with respect to worker proficiencies (and resulting stochastic realizations of their outcomes).

It is immediate that in our setting if there is no adversary and no capacity constraints for the workers, all tasks should be assigned to the worker with the highest $p_w$.
Our focus, however, is how to optimally assign workers to tasks when there is an intelligent adversary who may subsequently (to the assignment) attack a set of workers.
In particular, we assume that there is an adversary (attacker) with the goal of minimizing the defender's utility $u_{def}$; thus, the game is zero-sum.
To this end, the attacker chooses a set of $ \tau $ workers to attack, for example, by deploying a cyber attack against the corresponding computer nodes, physical attacks on search and rescue robots, or attacks against the devices on which the human workers performs their tasks.
Alternatively, our goal is to be robust to $ \tau $-worker failures (e.g., $N$-$\tau$-robustness~\cite{Chen14}).
We encode the attacker's strategy by a vector $\alpha$ where $\alpha_w = 1$ iff a worker $w$ is attacked (and $\sum_w \alpha_w = \tau$ since $ \tau $ workers are attacked).
The adversary's attack 
takes place \emph{after} the tasks have already been assigned to workers, where the attacker knows the actual assignments of tasks to workers before deploying the attack, and the consequence of an attack on a worker $w$ is that all tasks assigned to $w$ fail to be successfully completed.
Clearly, when an attacker is present, the policy of assigning all tasks to the most competent worker (when there are no capacity constraints) will yield zero utility for the defender, as the attacker will simply attack the worker to whom all the tasks are assigned.
The challenge of how to split the tasks up among workers, trading off quality with robustness to attacks, is the subject of our inqury.
Formally, we aim to compute a strong Stackelberg equilibrium of the game between the defender (leader), who chooses a task-to-worker assignment policy, and the attacker (follower), who attacks a single worker~\cite{stackelberg1952theory}.

\section{Homogeneous tasks}
We start by considering tasks which are \emph{homogeneous}, that is, $u_t = u_{t'}$ for any two tasks $t,t'$.
Without loss of generality, suppose that all $u_t = 1$. 
Note that since all tasks share the same utility, if $ B < m $, the defender is indifferent regarding the identity of tasks being assigned.
Further, it is immediate that we never wish to waste budget, since assigning a worker always results in non-negative marginal utility.
Consequently, we can simply randomly subsample $B$ tasks from the set of all tasks, and consider the problem with $m=B$.

We overload the notation and use $s = \{s_1,\ldots,s_n\}$ to denote the number of tasks allocated to each worker.
Although the space of deterministic assignments is large, we now observe several properties of optimal assignments which allow us to devise an efficient algorithm for this problem. 

\begin{prop}\label{prop:oneEachDeterministic} 
	Suppose that tasks are homogeneous. 
	For any assignment $s$ there is a weakly utility-improving assignment $s'$ for the defender which assigns each task to a single worker.
\end{prop} 
\begin{proof}
	Consider an assignment $s$ and the corresponding best response by the attacker, $\alpha$, in which a worker $\bar{w}$ is attacked.
	Let a task $\bar{t}$ be assigned to a set of workers $W_{\bar{t}}$ with $|W_{\bar{t}}| = k > 2$.
	Then there must be another task $t'$ which is unassigned.
	Now consider a worker $w \in W_{\bar{t}}$.
	Since utility is additive, we can consider just the marginal utility of any worker $w'$ to the defender and attacker; denote this by $u_{w'}$.
	Let $T_{w'}$ be the set of tasks assigned to a worker $w'$ under $s$.
	Let $u_w = \sum_{t \in T_w} u_{wt}^M$, where $u_{wt}^M = u_{t} \Pr\{\delta(L_t(s)) = \iota_t\} - u_{t} \Pr\{\delta(L_t(s)\setminus L_t^w) = \iota_t\}$ is the marginal utility of worker of $w$ towards a task $t$.
	Clearly, $u_w \le u_{\bar{w}}$, since the attacker is playing a best response.
	
	Suppose that we reassign $w$ from $\bar{t}$ to $t'$.
	If $w = \bar{w}$, the attacker will still attack $w$ (since the utility of $w$ to the attacker can only increase), and the defender is indifferent.
	If $w \ne \bar{w}$, there are two cases: (a) the attacker still attacks $\bar{w}$ after the change, and (b) the attacker now switches to attack $w$.
	Suppose the attacker still attacks $\bar{w}$.
	The defender's net gain is $p_{w} - u_{w\bar{t}}^M \ge 0$.
	If, instead, the attacker now attacks $w$, the defender's net gain is $u_{\bar{w}} - u_w \ge 0$.
\end{proof}
Consequently, we can restrict the set of assignments to those which assign a single worker per task; we denote this restricted set of assignments by $S$.
Given a assignment $s \in S$ and the attack strategy $\alpha$, the defender's expected utility is: 
\begin{equation}\label{eq:pure}
u_{def}(s,\alpha)=\sum_{w\in W}s_{w} p_{w}(1-\alpha_w)
\end{equation}
Next, we show that
there is always an optimal assignment that assigns tasks to the $k$ most proficient workers, for some $k$.

\begin{prop}\label{th:optDet}
	In an optimal assignment $s$, suppose that $s_i > 0$ for $i > 1$.
	Then there must be an optimal assignment in which $s_{i-1} > 0$.
\end{prop}
\begin{proof}
	Consider an optimal assignment $s$ and the attacker's best response $\alpha$ in which $\bar{W}$ is the set of workers being attacked.
	Now, consider moving 1 task from $i$ to $i-1$. We denote the updated set of workers attacked (due to this change) as $ \bar{W'} $.
	Suppose that $i \in \bar{W}$, that is, the worker $i$ was initially attacked.
	If $ i-1 \in \bar{W'} $, there are two potions: 1) $ i \in \bar{W'} $ (i.e., $ i $ is still being attacked) and hence the net gain to the defender does not change, and 2) $ i \notin \bar{W'} $ and hence the net gain to the defender is $p_i (|T_i|-1) \ge 0$.
	If $ i-1 \notin \bar{W'}$, the net gain is $p_{i-1} > 0$.
	Suppose that $i \notin \bar{W}$.
	If $i-1$ is now attacked, the net gain is $p_w (|T_w|-1) \ge 0$ (where $ w \in \bar{W} $ and $ w \notin \bar{W'} $).
	Otherwise (i.e., $ i-1 \notin \bar{W'} $), the net gain is $p_{i-1} - p_i \ge 0$.
\end{proof}

We can now present an assignment algorithm for optimal assignment (Algorithm~\ref{alg:deterministic}) which has complexity $O(n^2m\tau)$.
The intuition behind the algorithm is to consider each worker $i$ as a potential target of an attack, and then compute the best assignment subject to a constraint that $i$ is attacked (i.e., that $p_i s_i \ge p_js_j$ for all other workers $j \ne i$).
Subject to this constraint, we consider all possible numbers of tasks that can be assigned to $i$, and then assign as many tasks as possible to the other workers in order of their proficiency (where the $ \tau $ workers that contribute the most to the defender's utility are attacked). The only special case (Steps 7-10) is when assigning the last worker. In this case, it may be beneficial to alternate the last two workers' assignments to result in a more beneficial overall assignment.
Optimality follows from the fact that we exhaustively search possible targets and allocation policies to these, and assign as many tasks as possible to the most effective workers.

\begin{algorithm}[hbt]
	\caption{Homogeneous assignment}
	\label{alg:deterministic} 
	\begin{flushleft}
		\textbf{input:} The set of workers $W$, and their proficiencies $P$\\
		\textbf{return:} The optimal policy $ s^* $\\
	\end{flushleft}
	\vspace{-5pt}
	\begin{algorithmic}[1]
		\State $ u_{max} \leftarrow 0 $
		\For {$ i \in \{1,\ldots,n\} $}
		\For {$ s_i \in \{1,\ldots, c_i\} $}
		\State $\Upsilon_i \leftarrow s_ip_i,$
		$ B \leftarrow m - s_i $
		\For { $ j \in \{1,\ldots,n\}\setminus i $ }
		\State $ s_j \leftarrow \min(\floor*{\frac{p_i}{p_j}  s_i},B,c_j), B \leftarrow B-s_j $
		\If {$j<n \land B+1 \leq \min(\floor*{\frac{p_i}{p_{j+1}}  s_i}-1,c_{j+1})$}
		\State $ s' \leftarrow s $, $ s'_j\leftarrow s_j -1 $
		\If {$u_{def}(s,\alpha)\leq u_{def}(s',\alpha')+p_{j+1}$}
		\State $ s_j \leftarrow s_j-1, B \leftarrow B + 1 $
		\EndIf
		\EndIf					
		\State $ \Upsilon_j \leftarrow s_jp_j $
		\EndFor
		\State Sort $ \Upsilon $ in ascending order
		\State $ util\leftarrow \sum_{k=1}^{n-\tau}\Upsilon_k $
		\If {$ util >  u_{max}$}
		\State $ u_{max} \leftarrow util $, $ s^* \leftarrow s $
		\EndIf
		\EndFor
		\EndFor\\
		\Return $ s^* $                    
	\end{algorithmic} 
\end{algorithm}

Next, we turn to show that 	Algorithm 1 computes an optimal SSE commitment when tasks are homogeneous.
For readability, we denote the worker associated with the highest utility as $ w^{max} $, i.e., $ s_{w^{max}}p_{w^{max}} \geq s_wp_w \forall w \in W $. We focus on the case where at least one worker is not attacked as otherwise all possible assignments result in $ u_{def} = 0 $.
\begin{prop}\label{prop:mostProficient}
	Given an assignment $ s $ resulted from Algorithm 1, changing (i.e., decreasing or increasing) the number of tasks allocated to the worker with the highest utility ($ w^{max} $) is  promised to result in the same of lower utility for the defender.
\end{prop}
\begin{proof}
	Since Algorithm 1 iterates over all possible assignments for this worker (Step 3), if assigning less (or more) tasks was profitable, this updated assignment was resulted by the algorithm. 
\end{proof}
\begin{prop}\label{prop:moreBecomesTarget}
	Given an assignment $ s $  and two (arbitrary) workers $ j $ and $ k $, such that $ p_k < p_j $ and $ \alpha_j = 0 $ (i.e., $ j $ is not attacked under $ s $), the defender cannot move a task from $ k $ to $ j $ without making $ j $ a target.
\end{prop}
\begin{proof}
	Assume in negation that $ j $ can be assigned with an additional task and that $ j \notin T $ still holds.
	On each iteration, Algorithm 1 assigns  worker $ i $ with $ s_i $ tasks (Step 3), than each time the algorithm gets to Step 5, it assigns some other worker with the maximal amount of tasks such that this worker do not contribute to the defender's utility more than $ i $ does. Specifically, this is also the case for $ j $. Hence, by assigning $ j $ an additional task it must become a target (i.e., $ j \in T $), contradicting the assumption. 
\end{proof}

\begin{theorem}
	Algorithm 1 computes an optimal SSE commitment when tasks are homogeneous.
\end{theorem}
\begin{proof}
	Assume in negation that there exists some assignment $ s'  \neq s$ and its corresponding attack strategy, $ \alpha' $, such that $ u_{def}(s',\alpha') > u_{def}(s,\alpha) $.
	Specifically, there exist $ 16 $ different ways to make a single change in an assignment (as detailed below). For each of these changes, we prove that such assignment is not possible or contradicts the above assumption (i.e., do not improve the defender's expected utility). For readability of the proof, we denote $ w_i^s $ as worker $ i $ under assignment $ s $.
	\begin{enumerate}
		\item \textbf{Move a task from a non-target to a non-target }$ w_1^s \notin T $, $ w_1^{s'}\notin T $, $ w_2^s\notin T $, $ w_2^{s'} \notin T $: If $ p_1 \geq p_2 $, this change can either leave the utility as is (if $ p_1 = p_2 $) or decrease it (assign a task to a less proficient worker), hence contradicting the assumption. Otherwise, if $ p_1 < p_2 $, according to Proposition \ref{prop:moreBecomesTarget}, a more proficient non-target cannot be assigned with additional task and remains non-target. 
		\item $ w_1^s \notin T $, $ w_1^{s'}\notin T $, $ w_2^s\notin T $, $ w_2^{s'} \in T $:  Moving a task from a non-target worker to a less proficient non-target (i.e., if $ p_1 \geq p_2$) cannot make the less proficient worker a target ($ w_2^{s'} \in T $ is not possible in this case). Otherwise, if $ p_1 < p_2 $, $ p_2 $ will become  the worker associated with the highest utility under $ s' $ following this change. Still, if this assignment resulted in an higher utility for the defender, this was the output of Algorithm 1. Since this is not the output, $ u_{def}(s,\alpha) \geq u_{def}(s',\alpha) $, contradicting the assumption.
		
		\item $ w_1^s \notin T $, $ w_1^{s'}\notin T $, $ w_2^s\in T $, $ w_2^{s'} \notin T $: A worker that is currently being attacked cannot be assigned with an additional task and not be attacked anymore.
		\item $ w_1^s \notin T $, $ w_1^{s'}\notin T $, $ w_2^s\in T $, $ w_2^{s'} \in T $: Moving a task from a non-target worker to a target will only reduce the utility of the defender, contradicting the assumption.
		\item $ w_1^s \notin T $, $ w_1^{s'}\in T $, $ w_2^s\notin T $, $ w_2^{s'} \notin T $: A non-target worker cannot give a task and become a target.
		\item $ w_1^s \notin T $, $ w_1^{s'}\in T $, $ w_2^s\notin T $, $ w_2^{s'} \in T $: A non-target worker cannot give a task and become a target.
		\item $ w_1^s \notin T $, $ w_1^{s'}\in T $, $ w_2^s\in T $, $ w_2^{s'} \notin T $: A non-target worker cannot give a task and become a target.
		\item $ w_1^s \notin T $, $ w_1^{s'}\in T $, $ w_2^s\in T $, $ w_2^{s'} \in T $: A non-target worker cannot give a task and become a target.
		\item $ w_1^s \in T $, $ w_1^{s'}\notin T $, $ w_2^s\notin T $, $ w_2^{s'} \notin T $: Since $ w_2 $ assigned with an additional task and still not attacked, it must be the least proficient assigned worker. Thus, the new target (instead of $ w_1 $) is some other worker that results in the highest utility for the defender. According to Proposition~\ref{prop:mostProficient}, if this step was beneficial, it was resulted by Algorithm 1 that considers each worker as $ w^{max} $.
		\item $ w_1^s \in T $, $ w_1^{s'}\notin T $, $ w_2^s\notin T $, $ w_2^{s'} \in T $: If $ p_2 \geq p_1 $ this change make $ w_2 $ the worker who contributes most. According to Proposition~\ref{prop:mostProficient}, if this step was beneficial, it was resulted by Algorithm 1 that considers each worker as $ w^{max} $.
		Otherwise, if $ p_1 > p_2$, $ s_1  < s_2 $. The gain from this switch is $ (s_1-1)p_1-s_2p_2 $. This gain will be positive only if $ s_1p_1-p_1 > s_2p_2 $. Still, this is only possible if $ w_2 $ is the least proficient worker assigned (the difference of any other worker $ k $ from being a target is at most $ p_k $). If this is the least proficient assigned worker and becomes the target, it implies that  $ w_2 $ was assigned with the maximal amount of tasks such that it is not a target under $ s $. Hence, $ s2p2> s_1p_1-p_1 $, contradicting the assumption.
		\item $ w_1^s \in T $, $ w_1^{s'}\notin T $, $ w_2^s\in T $, $ w_2^{s'} \notin T $: A worker that is currently being attacked cannot be assigned with an additional task and not be attacked anymore.
		\item $ w_1^s \in T $, $ w_1^{s'}\notin T $, $ w_2^s\in T $, $ w_2^{s'} \in T $: Since $ w_1 $ is no longer a target (and $ w_2 $ was initially a target), there exists some other worker $ w' $ that prior to the change was not a target but becomes one due to the change (i.e, currently contributes more than $ w_1^{s'} $). Hence, the defender's utility decreases due to this change, contradicting the assumption. 
		\item $ w_1^s \in T $, $ w_1^{s'}\in T $, $ w_2^s\notin T $, $ w_2^{s'} \notin T $: Assigning another task to a non-target and keeping it a non-target is only possible if  $ w_2 $ is the least proficient worker assigned (otherwise, this worker will become a target). Still, the difference between the utility resulted from each assigned worker $ k $ (beside maybe the least proficient worker) and the worker who contributes the most, $ w^{max} $, is at most $ p_k $. By reducing the number of tasks assigned to $ w_1 $, its difference from $ w^{max} $ becomes more than $ p_k $ for some worker $ k \in W $. Hence, there is no way that $ w_1 $ is assigned with one less task and still a target.
		\item $ w_1^s \in T $, $ w_1^{s'}\in T $, $ w_2^s\notin T $, $ w_2^{s'} \in T $: $ w_2 $ becomes a target instead of some other worker. There are two possible cases: 1) $ w_2 $ is the least proficient worker assigned. Note that the difference between the utility resulted from each assigned worker $ k $ (beside maybe the least proficient worker) and the worker who contributes the most, $ w^{max} $, is at most $ p_k $. By reducing the number of tasks assigned to $ w_1 $, its difference from $ w^{max} $ becomes more than $ p_k $ for some worker $ k \in W $. Hence, there is no way that $ w_1 $ is assigned with one less task and still a target. 2) $ w_2 $ is some other worker. This means, that $ p_2 = w^{max} $ (under the new assignment). Still, if this assignment resulted in an higher utility for the defender, this was the output of Algorithm 1. Since this is not the output, $ u_{def}(s,\alpha) \geq u_{def}(s',\alpha) $, contradicting the assumption.
		\item $ w_1^s \in T $, $ w_1^{s'}\in T $, $ w_2^s\in T $, $ w_2^{s'} \notin T $: A worker that is currently being attacked cannot be assigned with an additional task and not be attacked anymore.
		\item $ w_1^s \in T $, $ w_1^{s'}\in T $, $ w_2^s\in T $, $ w_2^{s'} \in T $: Since both workers remain targets, $ u_{def}(s,\alpha) = u_{def}(s',\alpha') $, contradicting the assumption.
	\end{enumerate}
	Finally, since no single change is shown to be profitable from the defender's point of view, and any possible change in assignments can be represented as a set of single changes, we conclude that Algorithm 1 computes an optimal SSE commitment when tasks are homogeneous.
\end{proof}
\section{Heterogeneous tasks}
It turns out that the more general problem in which utilities are heterogeneous is considerably more challenging than the case of homogeneous allocation.
First, we show that even if the tasks' utilities are slightly different, it may be beneficial to assign the same task to multiple workers. Consider the case of an environment populated with 2 workers and 2 tasks. WLOG, we order the tasks by their utility, i.e., $u_{t_1} > u_{t_2}$. Regardless of the workers' proficiencies, assigning one worker per task will result in an expected utility of $\min (p_iu_{t_1}, p_ju_{t_2})$. On the other hand, assigning both workers to $t_1$ will result in an expected utility of $\min (p_iu_{t_1},p_ju_{t_1})$ which is always higher.
Aside from the considerably greater complexity challenge associated with solving problems with heterogeneous utilities suggested by this example, there is the additional challenge of incorporating (non-linear) decision rules into the optimization problem to resolving disagreement among workers, should it arise.
We begin by showing that if $B \leq m $, there is an optimal assignment in which only the $B$ tasks associated with the highest utility are included.
\begin{prop}
	Suppose that tasks are heterogeneous.
For any assignment $s$ there is a weakly utility-improving (i.e., results in the same or higher utility) assignment $s'$ for the defender which only assigns tasks from the set of  tasks with the $B$ highest utilities.	
\end{prop}
\begin{proof}
	For readability, we assume that tasks are ordered based on their utility in decreasing order (i.e., $u_i \geq u_j, \forall i \leq j$), and that a single worker is assigned per task; generalization is straightforward.
	Consider an assignment $s$ and the corresponding best response by the attacker, $\alpha$, in which the set of workers $\bar{W}$ is attacked.
	Let a task $t_i$ be s.t. $ i > B $.
	Then there must be another task $t_j$, s.t. $ j \leq B $, which is unassigned.
	Now consider a worker $w \in W_{t_i}$.
	Since utility is additive, we can consider just the marginal utility of any worker $w'$ to the defender and attacker; denote this by $u_{w'}$.
	Let $T_{w'}$ be the set of tasks assigned to a worker $w'$ under $s$.
	Let $u_w = \sum_{t \in T_w} u_{wt}^M$, where $u_{wt}^M$ is the marginal utility of worker of $w$ towards a task $t$.
	
	Suppose that we reassign $w$ from $t_i$ to $t_j$.
	If $w \in \bar{W}$, the attacker will still attack $w$ (since the utility of $w$ to the attacker can only increase), and the defender is indifferent.
	If $w \notin \bar{w}$, there are two cases: (a) the attacker still attacks $\bar{W}$ after the change, and (b) the attacker now switches to attack $w$.
	Suppose the attacker still attacks $\bar{W}$.
	The defender's net gain is $p_{w}u_j - u_{w\bar{t}}^M \ge 0$.
	If, instead, the attacker now attacks $w$, the defender's net gain is $u_{w'} - u_w \ge 0$. Where $ w' $ is the worker that is not being attacked anymore.
\end{proof}
This allows us to restrict attention to the $B$ highest-utility tasks, and assume that $m=B$.


We now show that the defender's assignment problem, denoted \textit{Heterogeneous tasks assignment (HTA)}, is NP-hard even if we restrict the strategies to assign only a single worker per task.
\begin{prop}
	HTA is strongly NP-hard even when we assign only one worker per task.
\end{prop}
\begin{proof}

	We prove the proposition by reducing the decision version of the \textit{Bin packing problem} (BP), which is a strongly NP-complete problem, to the decision version of the HTA problem. In the BP problem we are given a set $\{ o_1, o_2, ..., o_m \}$ of $m$ objects of sizes $\{ v_1, v_2, ..., v_m \}$ and a set of $n$ containers $\{ C_1, C_2, ..., C_n \}$, each of size $\gamma$, and we need to decide if all the objects can be fitted into the given containers. Our transformation maps the set of $m$ objects to a set of $m+1$ tasks $T = \{ t_1, t_2, ..., t_{m+1} \}$ with utilities $\{ v_1, v_2, ..., v_m, \gamma \}$ and the set of $n$ containers to a set of $n+1$ workers $W = \{ w_1, w_2, ..., w_{n+1} \}$. We consider the private case where all the workers have the same proficiency $p$ (i.e. $p_w = p, \forall w \in W$). The decision version of the HTA problem asks if there exists an assignment of the $m+1$ tasks to the $n+1$ workers that achieves a utility of at least $pV$, where $V = \sum_{i=1}^{m} v_i$.
	
	If we started with a YES instance of the BP problem, then there exists an assignment $\mathcal{A}$ that fits all $m$ objects into the $n$ containers. Consider the following assignment of tasks to workers in the HTA problem. If $\mathcal{A} (o_i) = C_j$, we assign task $t_i$ to worker $w_j$. Also, we assign task $t_{m+1}$ (with utility $\gamma$) to worker $w_{n+1}$. Note that no worker can achieve an individual utility greater than $p\gamma$, which is achieved by worker $w_{n+1}$. Thus, the utility of the overall task assignment is $\sum_{i=1}^{m} pv_i + p\gamma - p\gamma = pV$, meaning that our transformation produced a YES instance of the HTA problem.
	
	Now suppose that we ended up with a YES instance of the HTA problem. Then there exists a task assignment $\mathcal{B}$ such that the sum of utilities ($V^*$) minus the adversarial harm ($\gamma^*$) is at least $pV$ (i.e. $V^* - \gamma^* \geq pV$). Note that $V^* = \sum{i=1}^{m} pv_i + p\gamma = pV + p\gamma$ (each task is assigned to some worker). This implies $pV + p\gamma - \gamma^* \geq pV$ and $\gamma^* / p \leq \gamma$. Thus the utility sum (before performance p is applied) of the tasks assigned to any single worker cannot exceed $\gamma$. This could only happen if task $t_{m+1}$ (with utility $\gamma$) was the only task assigned to the corresponding player. WLOG let that worker be $w_{n+1}$. All other tasks must have been assigned to workers $\{ w_1, w_2, ..., w_n \}$. It is easy to see that this implies a feasible assignment of objects to containers in the BP problem - if $\mathcal{B} (t_j) = w_i$, for $1 \leq j \leq m$, then we place object $o_j$ in container $C_i$. Thus the transformation must have started off with a YES instance of the BP problem.
\end{proof} 

We now propose an algorithm which computes an approximately optimal assignment.
We begin by supposing that only one worker can be assigned per task (we relax this shortly).
In this case, the optimal attack can be computed using the following linear integer program:
\begin{subequations}
	\label{E:deterministicHet1}
	\begin{align}
	& \max_{\alpha} \sum_{w \in W} \alpha_w \sum_{t \in T} s_{wt} u_t  p_{w}  \label{objectiveDA}\\
	& s.t.: \sum_{w\in W} \alpha_w = \tau \label{cons:targets}\\
	&\alpha_w \in \{0,1\}.
	\end{align}
\end{subequations}
The objective (\ref{objectiveDA}) aims to maximize the effect of the attack (i.e., the utility of the targets). Constraint (\ref{cons:targets}) ensures that the adversary attacks exactly $ \tau $ workers.
First, note that the extreme points of the constraint set are integral, which means we can relax the integrality constraint to $\alpha_w \in [0,1]$.
In order to plug this optimization into the defender's optimal assignment problem, we convert this relaxed program to its dual form:
\begin{subequations}
	\label{E:deterministicHet2}
	\begin{align}
	& \min_{\lambda,\beta} \lambda\tau + \sum_w \beta_w\\
	& s.t.: \lambda + \beta_w \geq p_w\sum_{t \in T} s_{wt}u_t \ \forall \ w\\
          &  \beta \ge 0.
	\end{align}
\end{subequations}
Thus, the optimal assignment can be computed using the following linear integer program:
\begin{subequations}
	\label{E:deterministicHet}
	\begin{align}
	& \max_{s,\gamma,\lambda,\beta} \displaystyle\sum\limits_{w\in W}p_{w}\sum\limits_{t\in T} s_{wt} u_t   - \gamma  \label{objectiveDH}\\
	& s.t.: 
	\gamma  \geq \lambda\tau + \sum_w \beta_w\label{cons:maxDH}\\
	&\lambda + \beta_w \ge \sum\limits_{t\in T} s_{wt}u_t  p_{w}, \forall w \in W \label{cons:maxDHw}\\
	&\displaystyle\sum\limits_{w\in W}\sum\limits_{t\in T} s_{wt} = m \label{cons:allTasksDH}\\
	&\sum_w s_{wt} = 1, \forall t\in T\label{cons:singleWorkerTask}\\
	& \sum_t s_{wt} \leq c_w, \forall w\in W \label{cons:Capacity}\\
	&s_{wt} \in \{0,1\}.
	\end{align}
\end{subequations}
The objective (\ref{objectiveDH}) aims to maximize the defender's expected utility given the adversary's attack (second term). Constraint (\ref{cons:maxDH} and \ref{cons:maxDHw}) validates that the adversary's targets are the workers who contribute the most to the defender's expected utility and Constraint (\ref{cons:allTasksDH}) ensures that each allocation assigns all the possible tasks among the different workers.
Finally, Constraint~\eqref{cons:singleWorkerTask} ensures that only one worker is assigned for each task and Constraint~\eqref{cons:Capacity} ensures that no worker is assigned with more tasks than it can perform.




Next, we propose a greedy algorithm that attempts to incrementally improve utility by shifting workers among tasks, now allowing multiple workers to be assigned to a task.
Whenever more than one worker is assigned to a given task, the defender has to choose a deterministic mapping $ \delta $ to determine the outcome. 
We consider a very broad class of \emph{weighted majority} functions for this purpose (natural if successful completion of a task means that the worker returned the correct label).
In this mapping, each worker $ w $ is assigned a weight $ \theta_w $, and the final label is set according to the weighted majority rule, i.e., $ \delta(L_t) = \mathrm{sgn}(\sum_{w \in W_t(s)} \theta_w l_w) $.

In order to approximate the defender's expected utility, we use the sample average approximation (SAA)~\cite{kleywegt2002sample} for solving stochastic optimization problems by using Monte-Carlo simulation.
	Using this approach, the defender's utility can be approximated by:
	\begin{equation}\label{eq:HetDetMulti}
	u_{def}(C_K, W') = \sum_{t \in T}u_t \left(\sum_{k = 1}^{K} \frac{\mathbb{I}\{\mathrm{sgn}{\sum_{w \in W'}  s_{wt}\theta_w C_{wtk}}\}}{K}\right)
	\end{equation}
	where $ C_K $ is a set of $ K $ matrices, each of size $ n $ over $ m $. Each cell $ C_{wtk} $ is a randomly sample based on $ p_w $ represents whether or not the worker $w$ successfully completed the task.  That is, $ C_{wtk} =1$ if worker $ w $ successfully completed task $ t $, and $ C_{wtk} =0$ otherwise. In a similar manner, $ s_{wt} = 1 $ if worker $ w $ is assigned to task $ t $, and $ s_{wt} = 0 $ otherwise. 

	Algorithm~\ref{alg:HetroDetr} formally describes the computation of this assignment. 
	Given an optimal assignment extracted using the mixed-integer linear program in Equation~\eqref{E:deterministicHet}, we iteratively alternate over all tasks in ascending order based on their utility. For each task, we reassign the worker associated with this task to the most beneficial task. 
If this reassignment improves the defender's utility, we label it as beneficial (Steps 9 and 10). 
Finally, we commit to the reassignment that will maximize the defender's utility (Step 12).
\begin{algorithm}[hbt]
	\caption{Heterogeneous assignment} \label{alg:HetroDetr} 
	\begin{flushleft}
		\textbf{input:} The set of workers $W$, and their proficiencies $P$\\
		\textbf{return:} The heuristic deterministic allocation
	\end{flushleft}
	\begin{algorithmic}[1]
		\State Extract the optimal 1-worker allocation using Equation~\ref{E:deterministicHet}
		\State $ util \leftarrow u_{def}(C_K, \alpha)$
		\For {$ t \in \{1,\ldots,m\} $}
		\For {$ w \in \{1,\ldots,n\} $}
		\State $ \overline{t} = t $
		\If {$ s_{wt} = 1 $}
		\For {$ t' \in \{m,\ldots,t+1\} $}
		\State  $ s_{wt'} = 1 $,  $ s_{wt} = 0 $, Update $ \alpha $
		\If {$ u_{def}(C_K, \alpha) > util $}
		\State $ \overline{t} = t' $, $ util \leftarrow u_{def}(C_K, \alpha)$
		\EndIf	
		\State $ s_{wt'} = 0 $, $ s_{wt} = 1 $
		\EndFor
		\State  $ s_{wt} = 0 $, $ s_{w\overline{t}} = 1 $
		\EndIf
		\EndFor
		\EndFor
		\State \textbf{return} $ s $
	\end{algorithmic}
\end{algorithm}

\section{Experiments}
We now experimentally demonstrate the effectiveness of our proposed approaches.
Workers' proficiencies are sampled using two distributions: 
a uniform distribution over the $[0.5,1]$ interval and an exponential distribution with $ \mu = 0.25 $ where proficiencies are truncated to be in this interval for the latter.
We compare our adversarial assignment algorithms to three natural
baselines: \textit{Split-$ k $} and two versions of
\textit{Monte-Carlo} (involving random assignment of tasks to workers).
Specifically, for the Split-$ k $ method, we divide tasks equally
among the top $ k $ workers.\footnote{The remainder is assigned in an iterative way from the least proficient worker to the most proficient one.}
For the Monte-Carlo
approach, we consider a simple variant which randomly distributes
tasks among all the workers, denoted by \emph{Monte-Carlo}, and a
variant of this which randomly distributes the tasks among the top $
\ceil{\frac{n}{2}} $ workers, denoted by \emph{Top Monte-Carlo}. In both
cases, the assigned worker for each task is picked uniformly at random.

\vspace{-10pt}
\paragraph{Homogeneous Tasks}

We begin by considering homogeneous tasks.
For each experiment,  we take an average of 5,000 sample runs.

Figure~\ref{fig:Baselines} presents the results comparing our
algorithm to baselines for 50 workers and tasks.
As the figure shows, our algorithm outperforms the baselines,
and the gap becomes particularly pronounced as the number of targets
increases.
Moreover, there doesn't appear to be a qualitative difference between
uniform and exponential distribution in this regard.
\begin{figure}[h!]
	\centering
	\begin{subfigure}{0.5\linewidth}
	\centering
	\includegraphics[width=1\linewidth]{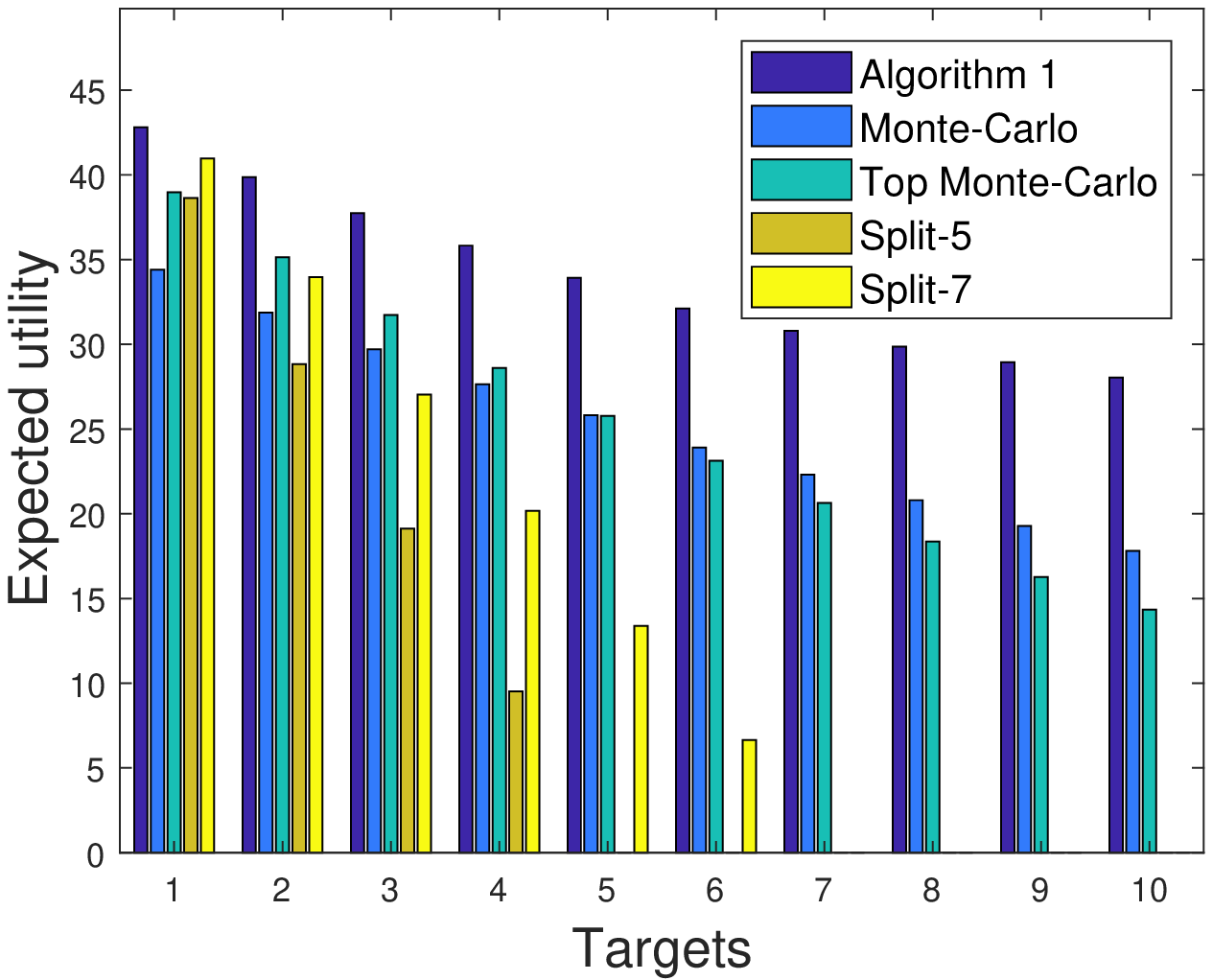}
	\caption{\small Uniform distribution}
	\label{fig:uniform1}
	\end{subfigure}%
	\begin{subfigure}{0.5\linewidth}
	\centering
	\includegraphics[width=1\linewidth]{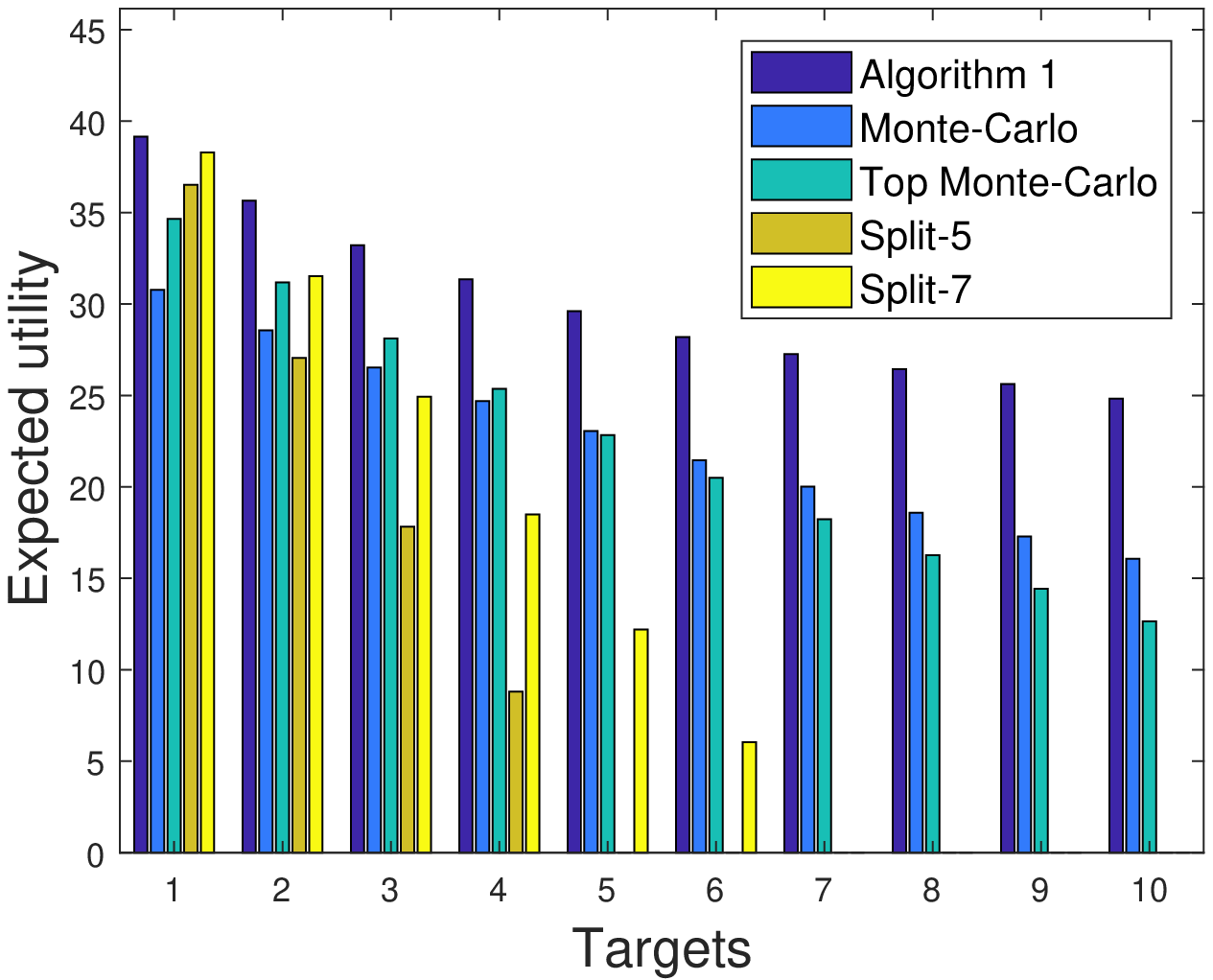}
	\caption{\small Exponential distribution}
	\label{fig:uniform2}
	\end{subfigure}%
	\caption{\small Homogeneous tasks: comparison to baseline
           methods.}
	\label{fig:Baselines}
\end{figure}

It is natural that we must trade off robustness with performance of
robust algorithms in non-adversarial settings.
We therefore conclude the homogeneous analysis by analyzing the loss incurred by allowing for robustness, compared to a solution which is optimal in non-adversarial settings.
We vary the number of workers from 2 to 50, and fix the number of
tasks at 100 and the number of targets optimized against at $ t = 1 $.
\begin{table}[h!]
	\centering
	\resizebox{\linewidth}{!}{\begin{tabular}{|c|c|c|c|c|c|c|c|c|c|c|}
			\hline 
			Workers& 5 &10&15&20&25&30&35&40&45&50\\\hline
			Exp. loss& 24.9\%&   17.4\%& 15.27\% &13.2\%& 11.6\%& 8.6\%& 5.8\%& 5.8\%& 6.5\%&    4.6\%\\ 
			\hline 
	\end{tabular} }
	\caption{Expected loss of using adversarial assignment in non-adversarial settings.}
	\label{tab:expectedLoss}
\end{table}


Table \ref{tab:expectedLoss} shows the expected loss of using 
adversarial task assignment in a non-adversarial settings. 
With only 5 workers, we pay a steep price (just under 25\%), but as
the number of workers increases, the loss shrinks; with 50 workers, we
only lose 4.6\% compared to optimal non-robust assignment.

\vspace{-10pt}
\paragraph{Heterogeneous Tasks}
We used CPLEX version 12.51 to solve the integer linear program above. 

First, we  analyze how the heterogeneous assignment given in
mixed-integer linear program (MILP)~\eqref{E:deterministicHet} performs compared to the baselines
when task utilities are sampled from $ U[0,1] $ and worker
proficiencies are samples from $ U[0.5,1] $. 
We use similar baseline methods to the ones used in studying
homogeneous task assignment.  

\begin{figure}[h!]
	\centering
	\begin{subfigure}{0.5\linewidth}	
	\includegraphics[width=1\linewidth]{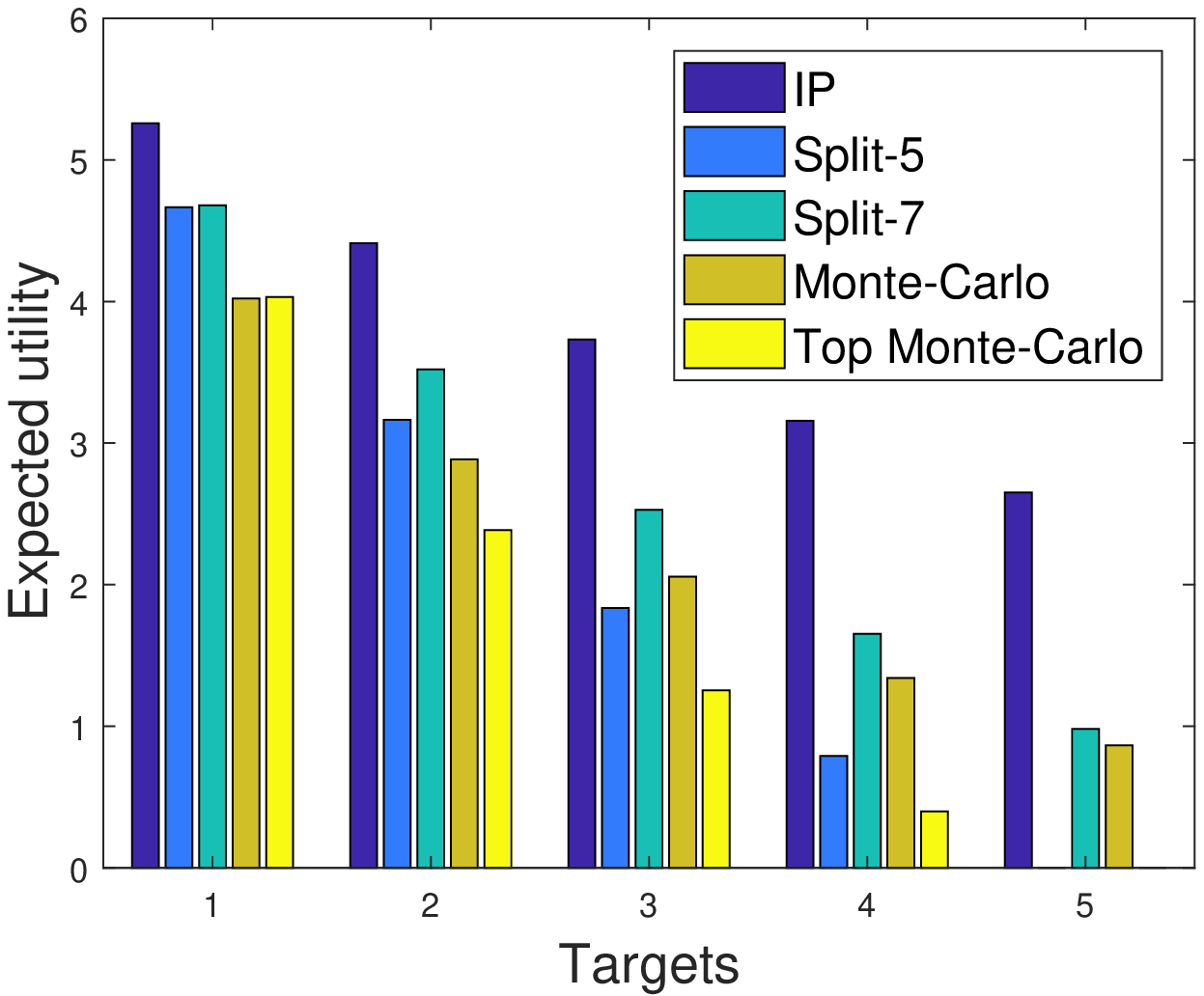}
	\caption{\small Uniform distribution}
	\end{subfigure}%
	\begin{subfigure}{0.5\linewidth}	
	\includegraphics[width=1\linewidth]{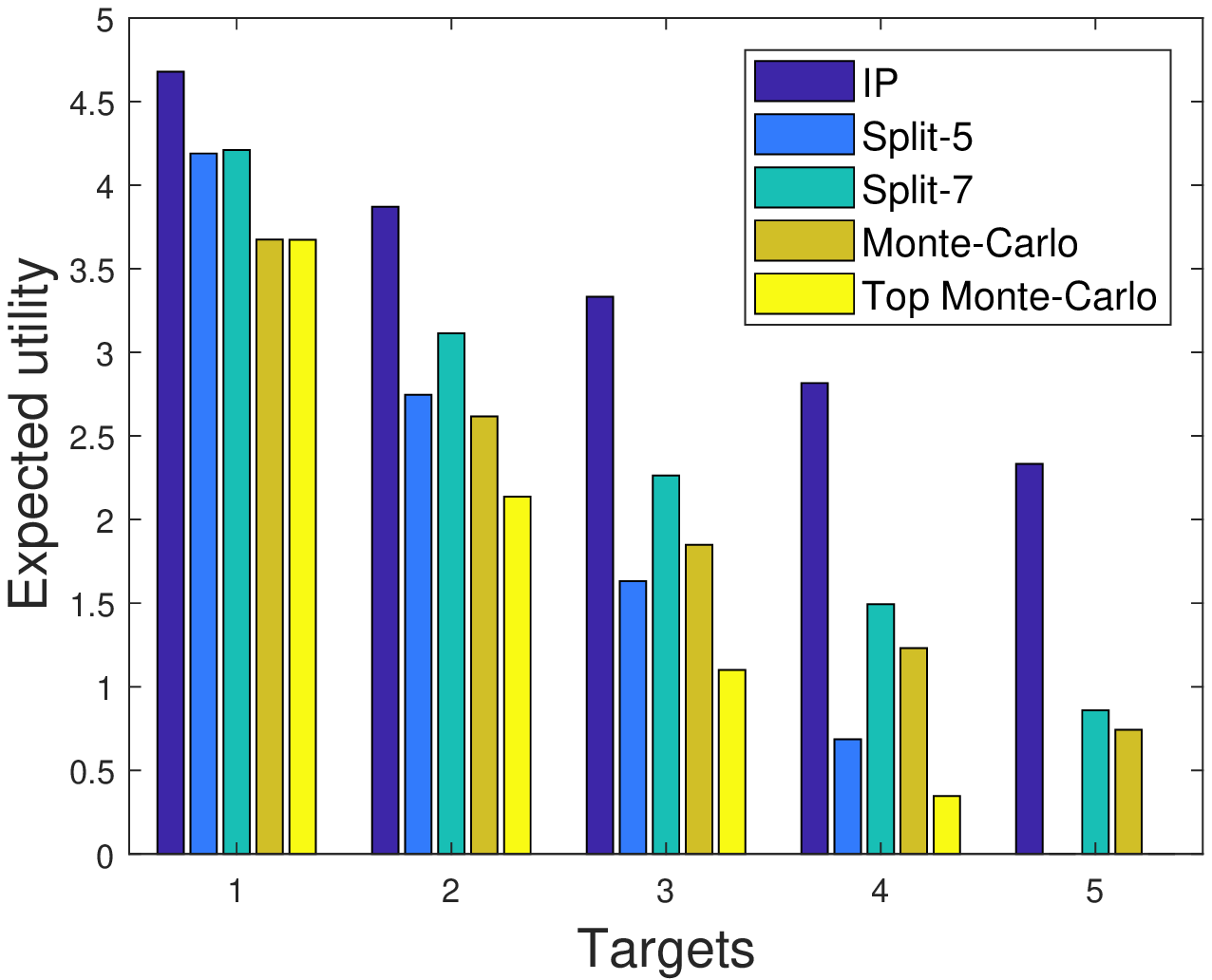}
	\caption{\small Exponential distribution}
	\end{subfigure}
	\caption{\small Heterogeneous tasks: comparison to baseline  methods.}
	\label{fig:baselineshet}
\end{figure}
Figure~\ref{fig:baselineshet} depicts the expected utility for the
defender when using each of the methods in an environment populated
with 15 tasks and 10 workers where the number of targets the adversary
attacks varies between $ 1 $ and $ 5 $ over $ 3,000 $ runs.
As is evidence from the figure, even the baseline mixed-integer linear
program (which assumes a single worker is assigned per task)
significantly outperforms the baselines, with the difference growing
as we increase the number of workers attacked.

\begin{table}[h!]
	\centering
		\resizebox{\linewidth}{!}{\footnotesize
	\begin{tabular}{|c|c|c|c|c|c|c|}
		\hline 
		&\multicolumn{6}{|c|}{Workers} \\ 
		\hline
		Dist. & Tasks& 2 &3&4&5&6\\\hline
		U[0,1] &3  & 57.08\% & 37.10\%    \\ 
		\cline{1-5}
		U[0,1] &4  & 26.47\%	 & 	9.88\%	 & 9.17\%    \\ 
		\cline{1-6}
		U[0,1] &5  & 22.03\% & 3.83\%  & 3.39\% & 3.46\% \\ 
		\hline 
		U[0,1] &6  & 19.98\% & 2\% & 1.79\% & 1.93\% & 1.66\% \\ \hline
		U[0,100] &3  &56.9\% & 37.92\%    \\ 
		\cline{1-5}
		U[0,100] &4  & 28.69\%	&9.59\%	&8.86\%    \\ 
		\cline{1-6}
		U[0,100] &5  & 20.02\%	&3.59\%	&3.51\%	&3.49\% \\ 
		\cline{1-7}
		U[0,100] &6  & 17.41\%&	1.59\%	&1.71\%&	1.64\%	&1.77\%\\
		\hline 
	\end{tabular}}
	\caption{Average improvement using
          Algorithm~\ref{alg:HetroDetr}; $ \tau = 1 $.}
	\label{tab:HetroDetImp}
\end{table}
\begin{table}[h!]
	\centering
		\resizebox{\linewidth}{!}{\footnotesize
	\begin{tabular}{|c|c|c|c|c|c|}
		\hline 
		&\multicolumn{5}{|c|}{Workers} \\ 
		\hline 
		Dist. & Tasks& 3&4&5&6\\\hline
		U[0,1] &3  & 1115.41\%     \\ 
		\cline{1-4}
		U[0,1] &4  & 46.27\%	 & 49.75\%    \\ 
		\cline{1-5}
		U[0,1] &5  & 19.52\% & 16.01\%  & 21.68\% \\ 
		\hline 
		U[0,1] &6  &9.88\% & 7.49\% & 10.9\% & 12.18\%   \\ \hline
		U[0,100] &3  &1130.13\%    \\ 
		\cline{1-4}
		U[0,100] &4  &58.23\%	&64.45\%    \\ 
		\cline{1-5}
		U[0,100] &5  & 17.97\%	&14.62\%	&21.21\% \\ 
		\cline{1-6}
		U[0,100] &6  & 8.62\%&	7.05\%	&9.83\%&	11.51\%\\
		\hline 
	\end{tabular}}
	\caption{Average improvement using
          Algorithm~\ref{alg:HetroDetr}; $ \tau = 2 $.}
	\label{tab:HetroDetImp2}
\end{table}


Next, we evaluate how much more we gain by using
Algorithm~\ref{alg:HetroDetr} after computing an initial assignment
using MILP~\eqref{E:deterministicHet}.
In these experimets we use a natural weighted majority decision rule
with $\theta_w = p_w$ (i.e., workers' proficiencies), and set
$K=2500$.
We consider two uniform distributions for this study: $U[0,1]$ and $U[0,100]$.
Each marginal improvement is averaged over 3,000
runs.

The results are shown in Tables \ref{tab:HetroDetImp} and
\ref{tab:HetroDetImp2}.
We can see that there are cases where assigning multiple workers per task can offer a significant benefit.
However, as the problem size increases, this benefit significantly
attenuates, and it may suffice to just rely on the assignment obtained
from the MILP.

\section{Conclusion}

We consider the problem of assigning tasks to workers when workers can be attacked, and their ability
to successfully complete assigned tasks compromised.
We show that the optimal assignment problem (in the sense of Stackelberg equilibrium commitment), when the attack takes place after the tasks have been assigned to workers, can be found in pseudo-polynomial time.
Furthermore, when tasks are heterogeneous, we show that the problem is more challenging, as it could be optimal to assign multiple workers to the same task.
Even if we constrain the assignment such that  only one worker is assigned per task, extracting the optimal assignment becomes strongly NP-Hard (we exhibit an integer linear program for the latter problem). Finally, we provide with an algorithm of converting this constraint assignment to one that allows multiple workers per task (and hence approximate optimal allocation). 
\section*{Acknowledgments}
This research was partially supported by the National  Science Foundation (CNS-1640624, IIS-1526860, IIS-1649972), Office of  Naval Research (N00014-15-1-2621),  Army  Research  Office (W911NF-16-1-0069), and National Institutes of Health (UH2 CA203708-01, R01HG006844-05).

\bibliographystyle{plain}
\bibliography{adversarialC}
\end{document}